\documentclass[12pt]{article}
\usepackage{amsmath,amssymb,amsthm}
\usepackage{fancyhdr}
\usepackage{hyperref} % No, this one has to be last!

\usepackage{calc}
\newcounter{hours}\newcounter{minutes}
\newcommand{\findtime}{%
    \setcounter{hours}{\time/60}%
    \setcounter{minutes}{\time-\value{hours}*60}}
\newcommand{\paddedminutes}{%
    \ifnum\value{minutes}<10 0\theminutes\else \theminutes\fi
    }

\newcommand{\theampmtime}{%
    \findtime
    \ifnum\value{hours}>12
        \setcounter{hours}{\value{hours}-12}%
        \thehours:\paddedminutes~pm%
        \findtime
    \else
        \ifnum\value{hours}=0
            12:\paddedminutes~am%
        \else
        \ifnum\value{hours}=12
            12:\paddedminutes~pm%
        \else
            \thehours:\paddedminutes~am%
        \fi
        \fi
    \fi}
\begin{document}

\newcommand{\creationtime}{\today \ \ @ \theampmtime}

\pagestyle{fancy}
\renewcommand{\headrulewidth}{0cm}
\chead{\footnotesize{Dougherty-Mauldin-Tiefenbruck}}
%\rhead{\footnotesize{\reviseddate}}
\rhead{\footnotesize{\creationtime}}
%\rhead{\footnotesize{\submitteddate}}
%\lhead{Version \version}
\lhead{}

\renewcommand{\qedsymbol}{$\blacksquare$} % the qed generated by \end{proof}

\newtheorem{theorem}              {Theorem}     [section]
\newtheorem{lemma}      [theorem] {Lemma}
\newtheorem{corollary}  [theorem] {Corollary}
\newtheorem{proposition}[theorem] {Proposition}
\newtheorem{algorithm}  [theorem] {Algorithm}
\newtheorem{conjecture} [theorem] {Conjecture}
\newtheorem{example}    [theorem] {Example}

\theoremstyle{definition}         
\newtheorem{definition} [theorem] {Definition}
\newtheorem*{claim}  {Claim}
\newtheorem*{notation}  {Notation}
\newtheorem*{remark}  {Remark}

\newcommand{\n}{m}
\newcommand{\wt}{\mathrm{wt}}
\newcommand{\x}{X}
\newcommand{\y}{Y}
\newcommand{\comppoly}[1]{{#1^c}}
\newcommand{\RM}{RM}
\newcommand{\coef}[1]{{\textrm{coef}(#1)}}
\newcommand{\ordinal}[1]{#1{\rm th}}
\newcommand{\BLpoly}[1]{{f_{#1}^{BL{*}}}}
\newcommand{\GFtwo}{\mathbf{F}_2}

%\begin{titlepage}

\setcounter{page}{0}

\title{The Covering Radius of the Reed--Muller Code $\RM(\n-4,\n)$
in $\RM(\n-3,\n)$
\thanks{This work was supported by the 
Institute for Defense Analyses.\newline
\indent \textbf{R. Dougherty}, \textbf{R. D. Mauldin}, and
\textbf{M. Tiefenbruck} are with the 
IDA Center for Communications Research, 
4320 Westerra Court, 
San Diego, CA 92121-1969 
(rdough@ccrwest.org, mauldin@ccrwest.org, mgtiefe@ccrwest.org).\newline
}}

\author{Randall Dougherty, R. Daniel Mauldin, and Mark Tiefenbruck\\ \ }

\date{
\textit{
%IEEE Transactions on Information Theory\\
%Version: \version\\
%Created: \creationtime \\
%Submitted: \submitteddate\\
%Revised:  \reviseddate 
Draft:   \creationtime 
% \today \\
%\Huge{Draft}
}}

\maketitle
\begin{abstract}
We present methods for computing the distance from a Boolean polynomial
on $\n$~variables of degree~$\n-3$ (i.e., a member of the Reed--Muller
code $\RM(\n-3,\n)$) to the space of lower-degree polynomials
($\RM(\n-4,\n)$). The methods give verifiable certificates for both
the lower and upper bounds on this distance.  By applying these
methods to representative lists of polynomials, we show that the
covering radius of $\RM(4,8)$ in $\RM(5,8)$ is 26 and the covering
radius of $\RM(5,9)$ in $\RM(6,9)$ is between 28 and 32 inclusive,
and we get improved lower bounds for higher~$\n$.  We also apply
our methods to various polynomials in the literature, thereby
improving the known bounds on the distance from 2-resilient
polynomials to $\RM(\n-4,\n)$.

\end{abstract}

\thispagestyle{empty}
%\end{titlepage}

\newpage
\section{Introduction}

Let $\RM(r,\n)$ be the \ordinal{$r$} order Reed-Muller code
of length~$2^\n$, which is the set of truth tables of
all Boolean polynomials of degree at most~$r$ 
on $\n$~variables.  The Reed-Muller codes are one of the best-understood
families of codes and have favorable properties such as
being relatively easy to decode~\cite{MacWilliams-Sloane},
but many questions regarding their covering radii remain open.

We will be studying distances from members of $\RM(r,\n)$
to $\RM(r-1,\n)$,
where the distance between two truth tables is defined as the
number of inputs on which they differ.
The maximum distance from any element of $\RM(r,\n)$ to the closest
element of $\RM(r-1,\n)$
is the covering radius of $\RM(r-1,\n)$
in $\RM(r,\n)$.
We will concentrate on the case $r=\n-3$
in this paper.
This is the largest unsolved case; 
the easy cases $r=\n$ and $r=\n-1$ are in~\cite[Sec.~9.1]{Cohen-book},
and $r=\n-2$ is settled in~\cite{McLoughlin}.

Schatz~\cite{Schatz} showed that the covering radius
of $\RM(2,6)$ (and, in particular, the covering radius
of $\RM(2,6)$ in $\RM(3,6)$) is~18.  
Hou~\cite{Hou2} later gave a noncomputer proof of this result.
Wang, Tan, and Prabowo \cite{Wang-Tan-Prabowo}
showed that the covering radius of $\RM(3,7)$ in $\RM(4,7)$
is~20.
Hou~\cite{Hou2}
showed that the covering radius of $\RM(4,8)$ in $\RM(5,8)$
is at least~22.
McLoughlin~\cite{McLoughlin}
showed that the covering radius of $\RM(\n-4,\n)$ in $\RM(\n-3,\n)$
is at least $2\n+2$ for odd~$\n$ and $2\n$ for even~$\n$;
asymptotically the sphere-covering
bound~\cite{MacWilliams-Sloane} gives a quadratic lower bound
on this covering radius.

In this paper, we present
(in Sections~\ref{sec:initlower}--\ref{sec:certalgs})
a method for proving a lower bound on
the distance from~$f$ to $\RM(\n-4,\n)$, where $f$~is a given
member of $\RM(\n-3,\n)$.
The method can require substantial computation
to generate a proof but produces a ``certificate'' that can
be used to verify the lower bound with far less computation.
We also describe
(in Section~\ref{sec:upper})
simple methods for producing an upper bound
on the distance (by finding a member of $\RM(\n-4,\n)$ close to~$f$,
so the upper bound is also easily verifiable).
By applying these methods to suitable
representative lists of Boolean polynomials, we show
(in Sections~\ref{sec:78var}--\ref{sec:9var})
that the
covering radius of $\RM(4,8)$ in $\RM(5,8)$ is 26 and the covering
radius of $\RM(5,9)$ in $\RM(6,9)$ is between 28 and 32 inclusive.
(The certificates and other data will be made available online.)
We also show
(in Section~\ref{sec:morevar})
that certificates for polynomials on $\n$ variables
can be modified to give certificates for certain polynomials on
more than $\n$ variables, thus giving improved lower bounds
on the covering radius of $\RM(\n'-4,\n')$ in $\RM(\n'-3,\n')$
for higher~$\n'$.

Once we have the distance to $\RM(\n-4,\n)$ (or lower and upper bounds on
it) for each of the representative polynomials, we can get this information
for any polynomial in $\RM(\n-3,\n)$ by looking up which representative
polynomial it is equivalent to.  In Section~\ref{sec:crypto}, we use this
method to compute this distance for a number of polynomials from the
literature that have desirable cryptographic properties.  This lets us
give improved bounds on several values of the function
$\hat\rho(t,r,n)$, the maximum distance between
a $t$-resilient function and $\RM(r,n)$.

\section{Preliminaries}

An $\n$-variable Boolean function is a function from $\GFtwo^\n$
to $\GFtwo$, where $\GFtwo$ is the two-element field; we also refer to
such functions as truth tables.  Any $\n$-variable Boolean function has a
unique representation as an $\n$-variable Boolean polynomial, a polynomial
over $\GFtwo$ where no variable occurs to a power higher than~$1$
(because $\GFtwo$ satisfies the identity $\x^2=\x$).
The Reed-Muller code
$\RM(r,\n)$ is the set of truth tables of
all Boolean polynomials of degree at most~$r$ 
on $\n$~variables.

{\bf Note:} For the rest of this paper we will
not distinguish between a polynomial and its truth table.

The book of MacWilliams and Sloane~\cite{MacWilliams-Sloane}
is a good reference for basic results on Reed--Muller codes.
In addition to the book~\cite{Cohen-book},
papers giving general upper and/or lower bounds on the covering radii of
Reed--Muller codes include \cite{Carlet-lb}, \cite{Carlet-Mesnager},
\cite{Cohen-Litsyn}, and~\cite{McLoughlin}.

For a given Boolean polynomial~$f$ on $\n$~variables, the
\emph{weight} of~$f$ is the number of input vectors $v\in \GFtwo^\n$
such that $f(v) = 1$ (i.e., the number of 1s in the truth
table for~$f$).  The \emph{distance} from~$f$ to a set~$S$
is the minimum of the weights of polynomials $f+p$
(which is the same as $f-p$) for $p$ in~$S$.

The distance from~$f$ to $\RM(r-1,\n)$ is
also known as the nonlinearity of order $r-1$ of~$f$
or the minimum weight of the coset $f + \RM(r-1,\n)$.
We will use the term ``minimum weight'' here; in fact, we
will sometimes abuse terminology slightly by saying
``minimum weight of~$f$'' as an abbreviation for
``minimum weight of the coset $f + \RM(r-1,\n)$''.

We refer to polynomials of
degree $\n-3$ as cocubic polynomials; similarly, we will use the words
``coquadratic'' for degree $\n-2$ and ``coquartic'' for degree $\n-4$.
(We will not use the word ``colinear.'')

The most straightforward (but usually not the most efficient)
way to compute the covering radius of $\RM(r-1,\n)$ in
$\RM(r,\n)$ is to compute the minimum weight of every
coset $f+\RM(r-1,\n)$ where $f$ is homogeneous of degree~$r$.

The work here can be cut down considerably by considering
the action of the 
general linear group $GL(\n,2)$ on Boolean polynomials:
an element~$h$ of $GL(\n,2)$ sends polynomial~$f$
to $f \circ h$.  This action is linear on~$f$ and
preserves the degree of~$f$ and also the
weight of~$f$ (since $h$ just permutes the inputs to~$f$).
It follows that $h$ maps the coset $f+\RM(r-1,\n)$ to
the coset $(f \circ h)+\RM(r-1,\n)$ (although
$f\circ h$ may not be homogeneous
even if $f$ is), and the action preserves
minimum weights of cosets.  Therefore,
this action partitions the cosets $f+\RM(r-1,\n)$ into equivalence classes,
and it suffices
to compute the minimum weight of one representative coset
from each equivalence class.

Another useful operation is complementation: given a polynomial~$f$
on $\n$~variables that is homogeneous of degree~$r$, one can
get the complementary polynomial~$\comppoly{f}$ of degree $\n-r$
by replacing each $r$-variable monomial in~$f$ with the
product of the other $\n-r$ variables.  It turns out, as shown
in Hou~\cite{Hou}, that homogeneous degree-$r$ polynomials $f$ and~$g$
are in $GL(\n,2)$-equivalent cosets of $\RM(r-1,\n)$ if and
only if $\comppoly{f}$ and~$\comppoly{g}$ are in
$GL(\n,2)$-equivalent cosets of $\RM(\n-r-1,\n)$.
Hence, given a list of homogeneous cubic polynomials that
(together with 0) are representatives for the
$GL(\n,2)$-equivalence classes of $\RM(3,\n)/\RM(2,\n)$,
we can complement these polynomials to get a list of homogeneous
cocubic polynomials that
(together with 0) are representatives for the
$GL(\n,2)$-equivalence classes of $\RM(\n-3,\n)/\RM(\n-4,\n)$.

Hou~\cite{Hou} gives a list of representative cubics in
up to 8 variables (six in up to 6~variables, six more in 7~variables,
and twenty more in 8~variables).
Brier and Langevin~\cite{Brier-Langevin} produced a list
of 349 representative 9-variable cubics, which is available from
Langevin's web site~\cite{Brier-Langevin-web}.
We can complement the polynomials in these lists to get
lists of representative cocubic polynomials.  (In the case
of 6~variables, where cubic is the same as cocubic,
we will not bother to complement.)

Two more facts will be useful to us later.
One is that an $\n$-variable Boolean polynomial has degree $\n$
if and only if its weight is
odd~\cite[Chapter~13, Problem~5]{MacWilliams-Sloane}.
The other is the following proposition.

\begin{proposition}
\label{prop:multvar}
If the $\n$-variable Boolean polynomial $p$ is such that the coset\/
$p+\RM(r,\n)$ has minimum weight $w$,
and\/ $\y$ is a new variable, then
the coset\/ $\y p+\RM(r+1,\n+1)$ has minimum weight $w$.
\end{proposition}

This is included in the proof of \cite[Lemma~9.3.1]{Cohen-book},
but we provide a proof here for convenience.

\begin{proof}
Let $\wt_\n(q)$ denote the weight of a Boolean polynomial~$q$ on
$\n$~variables.
If $p'$ is in $p+\RM(r,\n)$ and has weight $w$, then
$\y p'$ is in $\y p+\RM(r+1,\n+1)$ and has weight $w$.
It remains to show that any polynomial in $\y p+\RM(r+1,\n+1)$
has weight at least~$w$.  Any such polynomial has the form
$\y p + \y q + q'$ where $q$ and~$q'$ are polynomials in the
original $\n$~variables, $q$~has degree at most~$r$, and $q'$~has
degree at most~$r+1$.  There are
$\wt_\n(q')$ solutions to $\y p + \y q + q' = 1$ with $\y = 0$, and
$\wt_\n(p + q + q')$ solutions with $\y = 1$.  So
$\wt_{\n+1}(\y p + \y q + q') = \wt_\n(p + q + q') + \wt_\n(q')
\ge \wt_\n(p + q + q' + q') = \wt_\n(p + q) \ge w$,
as desired.
\end{proof}

\section{Initial lower bound results}
\label{sec:initlower}

We start by showing that most cocubic Boolean polynomials on $\n$~variables
have weight at least $2\n+2$.

\begin{theorem}
\label{thm:weight}
Any Boolean polynomial~$f$ of degree at most $\n-3$ in $\n$~variables
that has no affine (degree-1) factors
has weight at least $2\n+2$.
\end{theorem}

\begin{proof}
For any such $f$, construct an $(\n+1)\times 2^{\n}$ matrix~$T$
over $\GFtwo$ as follows.
We label the rows with the polynomials
$f, \x_1f, \x_2f,\dots,\x_\n f$, and we label the columns with
the $2^\n$ possible settings for the input variables $\x_1,\dots,\x_\n$;
the entry for row $g$ and column $v$ is the value $g(v)$.
(So this is a joint truth table for the listed polynomials.)

We first show that the rows of $T$ are linearly independent; this
is equivalent to saying that the label polynomials are linearly
independent.  Suppose the label polynomials are not linearly independent;
then some nonempty subset of them has sum~0, which means that
$af = 0$ where $a$ is a nonzero polynomial of degree at most~1.
We cannot have $a = 1$, since this would give $f=0$, and the polynomial~$0$
has affine factors.  So $a$ has degree 1, and so does $a+1$, and we have
$(a+1)f = af+f = f$, so $f$ has an affine factor, contradicting the
assumptions of the theorem.

Next, we show that the rows of~$T$ are orthogonal to themselves and
to each other.  The dot product of truth table rows labeled $g_1$ and
$g_2$ is the number of input vectors $v$ such that $g_1(v) = g_2(v) = 1$,
reduced modulo~2; this is the same as the weight of $g_1g_2$ modulo~2.
But the product of any two of our label polynomials (distinct or not)
is a polynomial of the form $f$, $\x_if$, or $\x_i\x_jf$, since $f^2=f$;
all of these polynomials have degree less than~$\n$ and hence
have even weight.  So all dot products of rows of~$T$ are~$0$.

This shows that the row space of~$T$ has dimension $\n+1$ and
codimension at least~$\n+1$, so the dimension of the ambient space
(i.e., the number of columns of~$T$) is at least $2\n+2$.

Now let $T'$ be $T$ with its all-0 columns (those labeled $v$
where $f(v)=0$) deleted.  The all-0 columns do not affect the
independence or orthogonality of rows, so the arguments above
show that the number of columns of~$T'$ is at least $2\n+2$.
The columns of~$T'$ are labeled with~$v$ such that $f(v)=1$,
so the number of columns of~$T'$ is the weight of~$f$,
and we are done.
\end{proof}

This result gives a lower bound on the weight of a cocubic polynomial $f$,
but we are interested in the minimum weight of $f$.
Thus, we will have to consider polynomials $f+p$ where $p$ has
degree at most $\n-4$.  In this case, instead of trying to show
that all such polynomials $f+p$ have no affine factor, it will be
more convenient to impose a linear independence assumption directly
(especially since verifying that assumption is probably the most
convenient way to prove that $f$~has no affine factor anyway).

\begin{theorem}
\label{thm:minwt}
If a Boolean polynomial~$f$ of degree $\n-3$ in $\n$ variables
has the property that the degree-$(\n{-}2)$ parts of the products $\x_if$
($i=1,2,\dots,\n$) are linearly independent,
then $f+\RM(\n-4,\n)$~has minimum weight at least $2\n+2$.
\end{theorem}

\begin{proof}
We want to show that, for any polynomial~$p$ of degree at most $\n-4$,
the polynomial $f+p$ has weight at least $2\n+2$.  So construct
the matrices $T$ and $T'$ as in the proof of Theorem~\ref{thm:weight},
except we now label the rows $f+p$ and $\x_i(f+p)$ for
$i=1,2,\dots,\n$.

For any sum of one or more of the label polynomials that
includes at least one polynomial $\x_i(f+p)$, the coquadratic part
of this sum (which just involves $f$, since $\x_ip$ has degree
less than $\n-2$) is nonzero by our assumption on $f$, so the sum
is nonzero.  If a sum of one or more label polynomials does not
include any of the polynomials $\x_i(f+p)$, then the sum is
just $f+p$, which has degree $\n-3$ and is therefore nonzero.
Therefore, the label polynomials, and hence the rows of~$T$,
are linearly independent.

Any product of two of the label polynomials (distinct or not)
has the form $f+p$, $\x_i(f+p)$, or $\x_i\x_j(f+p)$; all of
these have degree less than~$\n$, so the rows of~$T$ are orthogonal
to themselves and each other.

Deleting the all-zero columns of~$T$ to produce $T'$ does not
affect either of these conclusions, so as before, the number
of columns of~$T'$ is at least twice the number of rows of~$T'$.
Therefore, the weight of~$f+p$ is at least $2\n+2$, as desired.
\end{proof}

To get lower bounds better than $2\n+2$ requires an extension of this
method, which we will demonstrate in the next section on
6-variable polynomials.

\section{The 6-variable case}
\label{sec:6var}

As we saw earlier, in order to compute the covering radius of
$\RM(2,6)$ in $\RM(3,6)$, we only need to compute the minimum
weights of a list of representative polynomials from which
every polynomial in $\RM(3,6)$ can be obtained by linear
transformations of the variables and/or adding terms of degree
less than~$3$.

Hou~\cite{Hou} lists
a complete set of
six representative cubic polynomials on
6~variables, but three of them essentially reduce
(by use of Proposition~\ref{prop:multvar}) to fewer variables:

\begin{itemize}
\item $f_1 = 0$: The zero polynomial on any number of variables
obviously has weight 0 and minimum weight 0.
\item $f_2 = \x_1\x_2\x_3$: This reduces to the case of the constant
polynomial~1 on three variables, which has weight 8 and
minimum weight 8.
\item $f_3 = \x_1\x_2\x_3+\x_2\x_4\x_5$: This reduces to the case of the
polynomial $\x_1\x_3+\x_4\x_5$ on five variables, which has weight 12
and (by Theorem~\ref{thm:minwt}) minimum weight 12.
\end{itemize}

We now consider the remaining three representative cubics.  First, let
$$f = f_4 = \x_1\x_2\x_3 + \x_4\x_5\x_6.$$
One can verify directly that $f$~has weight 14; by Theorem~\ref{thm:minwt}, $f$~has minimum weight 14.

We next look at
$$f = f_5 = \x_1\x_2\x_3 + \x_2\x_4\x_5 + \x_3\x_4\x_6.$$
This has weight 16 by direct computation.
Theorem~\ref{thm:minwt} gives a lower bound
of~14 for its minimum weight; to improve the lower bound to~16,
we will add one more row to the truth table in the proof of
Theorem~\ref{thm:weight}.

Let us list the degree-$4$ terms from $\x_if$ for this case:
\begin{align*}
   \x_1f:&\quad \x_1\x_2\x_4\x_5 + \x_1\x_3\x_4\x_6
\\ \x_2f:&\quad \x_2\x_3\x_4\x_6
\\ \x_3f:&\quad \x_2\x_3\x_4\x_5
\\ \x_4f:&\quad \x_1\x_2\x_3\x_4
\\ \x_5f:&\quad \x_1\x_2\x_3\x_5 + \x_3\x_4\x_5\x_6
\\ \x_6f:&\quad \x_1\x_2\x_3\x_6 + \x_2\x_4\x_5\x_6
\end{align*}
All of the monomials here are distinct, so these coquadratic
polynomials are linearly independent.

We now add to the truth table a new row given by $(\x_1\x_2+\x_4\x_6)f$.

First note that, when one multiplies $\x_1\x_2+\x_4\x_6$ by $f$,
the only degree-5 term that occurs is $\x_1\x_2\x_3\x_4\x_6$,
and it occurs twice and hence cancels out, so
$(\x_1\x_2+\x_4\x_6)f$ has degree~$4$.
This is important because it implies that
the truth table row for $(\x_1\x_2+\x_4\x_6)f$ is orthogonal to
itself and to the rows for $f$ and $\x_if$.

The degree-4 part of $(\x_1\x_2+\x_4\x_6)f$ is
$\x_1\x_2\x_4\x_5 + \x_2\x_4\x_5\x_6$, which is linearly independent
of the degree-4 parts of $\x_if$ listed above.  A quick way
to verify this is to note that, if $\coef{rstu}$ is the coefficient
of the monomial $\x_r\x_s\x_t\x_u$ for a polynomial we are considering,
then $\coef{1245}+\coef{1346}$ is 0 (in $\GFtwo$) for
all of the $\x_if$ but is 1 for $(\x_1\x_2+\x_4\x_6)f$.

Putting these facts together, we can follow
the argument of Theorem~\ref{thm:weight} using the matrix~$T$
augmented with the new row for~$(\x_1\x_2+\x_4\x_6)f$:
the augmented matrix has a row space of dimension~8 and
codimension at least~8, so it must have at least~16 nonempty columns.
This shows (again) that the weight of~$f$ is at least 16.

However, we cannot yet conclude that the minimum weight of~$f$
is at least~16.  For Theorem~\ref{thm:minwt}, the addition of
lower-degree terms (degree at most $\n-4$) to~$f$ did not
affect the arguments; this is no longer the case now that
we are multiplying~$f$ by a quadratic.

So suppose we instead look at
\begin{equation}
\label{eq:g}
g = f + \sum_{r{<}s} c_{rs}\x_r\x_s + \text{lower-degree terms},
\end{equation}
where the lower-degree terms (degree at most~1) will have
no effect on the following arguments.

How much of the preceding proof still goes through?  The
degree-4 parts of $\x_ig$ are the same as those of $\x_if$,
so these are still linearly independent.  The product
$(\x_1\x_2+\x_4\x_6)g$ still has degree at most 4, which takes
care of the orthogonality conditions.

This leaves the proof of linear independence, and again we use
the coefficient combination $\coef{1245}+\coef{1346}$ to
prove this.  It turns out that, for $(\x_1\x_2+\x_4\x_6)g$,
we have
$$\coef{1245}+\coef{1346} = 1 + c_{13} + c_{45}.$$
This is not general enough to handle all possibilities for~$g$,
but we at least have a partial result: if $g$~is as in~\eqref{eq:g}
with $c_{13} + c_{45} = 0$, then the weight of~$g$ is at least~16.

It turns out we can handle the case $c_{13} + c_{45} = 1$ by repeating
the entire argument with $T$ augmented by the row for
$(\x_1\x_3+\x_4\x_5)g$ instead of $(\x_1\x_2+\x_4\x_6)g$.
We again find that
$(\x_1\x_3+\x_4\x_5)f$ has degree~4, so $(\x_1\x_3+\x_4\x_5)g$
has degree at most~4 and the orthogonality conditions are met.
Of course, the $\x_ig$ are still linearly independent, so that leaves
the question of $(\x_1\x_3+\x_4\x_5)g$'s independence.
It turns out that we can show independence for $(\x_1\x_3+\x_4\x_5)g$
by looking at the monomial coefficient $\coef{1345}$.
This comes out to 0 for $\x_ig$ and
$c_{13} + c_{45}$ for $(\x_1\x_3+\x_4\x_5)g$; since we are assuming
$c_{13} + c_{45} = 1$, we have the desired independence, so
the weight of~$g$ is at least~16.

Putting these two cases together to handle all possible $g$, we see
that the minimum weight of~$f$ is at least (and hence exactly) 16.

Finally, consider
$$f = f_6 = \x_1\x_2\x_3 + \x_1\x_4\x_5 + \x_2\x_4\x_6 + \x_3\x_5\x_6 +
         \x_4\x_5\x_6.$$
Up to renaming variables, this is the polynomial shown by
Schatz~\cite{Schatz} to have minimum weight 18.  

In order to prove a lower bound of 18 on the minimum weight of~$f$,
we need to add \emph{two} rows to the truth table in Theorem~\ref{thm:minwt}.
The argument is similar in flavor to the $f_5$ case, only more complicated,
so we will omit some of the easy steps.
We again must handle all $g$ of the form~\eqref{eq:g}.

We begin by considering $T$ augmented by the pair of rows for the polynomials
$(\x_1\x_5+\x_2\x_6)g$ and $(\x_1\x_2+\x_1\x_5+\x_5\x_6)g$.
We first check that the products $(\x_1\x_5+\x_2\x_6)g$ and
$(\x_1\x_2+\x_1\x_5+\x_5\x_6)g$ have degree at most 4; this shows that
the two corresponding rows of the truth table are
self-orthogonal and orthogonal to all of the $\x_ig$ rows.
We still have to verify that these two rows
are orthogonal to each other.

For independence, we first check that the $\x_ig$ rows
have linearly independent coquadratic parts. Next, we consider the
coefficient combinations $\coef{1245}+\coef{2356}$ and
$\coef{1235}+\coef{1245}+\coef{2456}$. Both of these are zero for the
$\x_ig$. We also find: 
\begin{itemize}
\item the degree-6 coefficient of
$(\x_1\x_5+\x_2\x_6)(\x_1\x_2+\x_1\x_5+\x_5\x_6)g$ is
$c_{34}$;
\item the combination $\coef{1245}+\coef{2356}$ for $(\x_1\x_5+\x_2\x_6)g$
is $1+c_{24}+c_{35}$;
\item the combination $\coef{1235}+\coef{1245}+\coef{2456}$ for
$(\x_1\x_5+\x_2\x_6)g$ is\newline $c_{23}+c_{24}+c_{45}$; and
\item the combination $\coef{1235}+\coef{1245}+\coef{2456}$ for
$(\x_1\x_2+\x_1\x_5+\x_5\x_6)g$ is $1+c_{23}+c_{35}+c_{45}$.
\end{itemize}
Thus, the two new rows will be orthogonal to each other
if $c_{34}=0$;
$(\x_1\x_5+\x_2\x_6)g$ will be independent of the $\x_ig$ if
$c_{24}+c_{35}=0$; and $(\x_1\x_2+\x_1\x_5+\x_5\x_6)g$ will be
independent of the $\x_ig$ and $(\x_1\x_5+\x_2\x_6)g$ if
$c_{23}+c_{35}+c_{45}=0$ and $c_{23}+c_{24}+c_{45}=0$.

Putting these pieces together and noting that
$$c_{23}+c_{35}+c_{45} = (c_{24}+c_{35}) + (c_{23}+c_{24}+c_{45}),$$
we see that $g$ must have weight at least 18 so long as
\begin{equation}
\label{eq:cond6a}
\begin{aligned}
         c_{34}&=0,
\\         c_{24}+c_{35}&=0,\text{ and}
\\         c_{23}+c_{24}+c_{45}&=0.
\end{aligned}
\end{equation}
To handle the remaining possibilities for $g$, we need to use
different quadratic multipliers and/or coefficient combinations.

For the next case, we may repeat the same arguments using the
quadratics $\x_1\x_6$ and $\x_1\x_5+\x_2\x_6$
and the coefficient combinations $\coef{1346}$
and $\coef{1234}+\coef{1345}+\coef{3456}$.
Then, we find:
\begin{itemize}
\item the degree-6 coefficient of
$\x_1\x_6(\x_1\x_5+\x_2\x_6)g$ is
$0$;
\item the coefficient $\coef{1346}$ for $\x_1\x_6g$
is $c_{34}$;
\item the combination $\coef{1234}+\coef{1345}+\coef{3456}$ for
$\x_1\x_6g$ is $0$; and
\item the combination $\coef{1234}+\coef{1345}+\coef{3456}$ for
$(\x_1\x_5+\x_2\x_6)g$ is $c_{34}$.
\end{itemize}
Hence, $g$ must have weight at least~18 if
\begin{equation}
\label{eq:cond6b}
c_{34}=1.
\end{equation}

A similar computation using quadratics $\x_1\x_6$ and $\x_1\x_5+\x_2\x_6$
(the same as the preced\-ing case)
and coefficient combinations $\coef{1236}+\coef{1246}+\coef{1456}$
and $\coef{1235}+\coef{1245}+\coef{2456}$
shows that $g$ must have weight at least~18 if it satisfies
\begin{equation}
\label{eq:cond6c}
c_{23}+c_{24}+c_{45}=1.
\end{equation}
Yet another such computation using quadratics $\x_1\x_6$ and
$\x_1\x_2+\x_1\x_5+\x_5\x_6$
and coefficient combinations $\coef{1246}+\coef{1356}$
and $\coef{1245}+\coef{2356}$
shows that $g$ must have weight at least~18 if
\begin{equation}
\label{eq:cond6d}
\begin{aligned}
         c_{24}+c_{35}&=1\text{ and}
\\         c_{23}+c_{24}+c_{45}&=0.
\end{aligned}
\end{equation}

It is easy to see that, no matter what~$g$ is, at least one of
the conditions \eqref{eq:cond6a}, \eqref{eq:cond6b}, \eqref{eq:cond6c},
and \eqref{eq:cond6d} must hold, so $g$ has weight at least 18.
Therefore, the minimum weight of~$f$ is 18.

The results in this section give a new noncomputer proof that the
covering radius of $\RM(2,6)$ in $\RM(3,6)$ is 18.

\section{Certificates}
\label{sec:certificates}

For a proof of the form given for $f_5$ and $f_6$ in the previous section,
it should be clear that \emph{verifying} such a proof may be much easier
than \emph{finding} it. While verifying those two proofs \emph{may} be
within grasp of a human, some of our later results are certainly not.
However, the skeptical reader does not simply have to trust us, as with
previous computational results mentioned herein. By conveying all the
necessary pieces for these proofs in a \emph{certificate}, the reader may
verify our claims with a relatively modest computation.

Here is a brief summary of the proofs. We start with a cocubic polynomial
$f$ in $\n$ variables, and we want to show that it has minimum weight $2\n+2+2k$ for some
$k$. We thus must show that every $g$ with the same cocubic terms as $f$
must have weight at least $2\n+2+2k$. Then, we divide the proof into a
number of cases. For each case, we choose $k$ quadratic multiples of $g$
and append their truth tables to $T$, as defined in
Theorem~\ref{thm:weight}. Provided the coefficients of $g$ satisfy
certain constraints, we can prove that all the rows of $T$ are linearly
independent and orthogonal, which establishes the bound on all such $g$.
When all of the cases, taken together, cover all possible $g$, our proof
is complete.

We now formalize the method by defining a \textit{level-$k$ certificate}
that contains all the necessary data to verify such a proof.

\begin{definition}
\label{def:cert}
Let $f$ be a Boolean polynomial of degree~$\n-3$ on $\n$~variables
$\x_1,\dots,\x_{\n}$.
A \textit{level-$k$ certificate} for~$f$ is a sequence of triples
$\langle C,q,r \rangle$ (which we will often refer to as ``subproofs'') where:
\begin{itemize}
\item $C$ is a list of affine equations on $\GFtwo$ variables
$c_M$ where $M$~is a monomial of degree $\n-4$
($C$~will be viewed as a condition to be assumed on the values
of these variables);
\item $q = [q_1,\dots,q_k]$ is a sequence of $k$ quadratics;
\item $r = [r_1,\dots,r_k]$ is a sequence of $k$ linear combinations of
coefficient specifiers $\coef{M}$ where $M$~is a monomial of degree $\n-2$;
\end{itemize}
meeting the following requirements:
\begin{enumerate}
\item
\label{item:indepoldold}
The degree-$(\n{-}2)$ parts of the polynomials $\x_1f,\dots,\x_{\n}f$
are linearly independent.  (This says nothing about the certificate;
it is a restriction on the polynomial~$f$.)
\item
\label{item:orthogoldnew}
For each triple $\langle C,q,r \rangle$ and each $j \le k$,
$q_jf$ has degree at most $\n-2$.
\item
\label{item:orthognewnew}
For each triple $\langle C,q,r \rangle$ and each pair $j' < j \le k$,
if condition~$C$ holds, then the coefficient of $\x_1\x_2\cdots\x_{\n}$
in $q_{j'}q_{j}(f+\sum_Mc_MM)$ is 0.
\item
\label{item:indepoldnew}
For each triple $\langle C,q,r \rangle$ and each $i \le \n$ and $j \le k$,
the coefficient combination~$r_j$ for~$\x_if$ is 0.
\item
\label{item:indepnewnew}
For each triple $\langle C,q,r \rangle$ and each pair $j' \le j \le k$,
if condition~$C$ holds, then the coefficient combination $r_j$
for $q_{j'}(f+\sum_Mc_MM)$ is 0 if $j'<j$ but is 1 if $j' = j$.
\item
\label{item:fcs}
The affine subspaces (flats) defined by the conditions~$C$ for all
triples $\langle C,q,r \rangle$ in the sequence cover the entire
space of possible values of the variables~$c_M$ (i.e.,
any assignment to these variables satisfies at least one
condition~$C$).
\end{enumerate}
\end{definition}

In particular, a simple level-0 certificate is just a single triple of
three null sequences; to verify that it is in fact a certificate
for~$f$, one just has to check requirement~\ref{item:indepoldold} above.
This is a certificate for~$f_4$.

Based on the arguments in the preceding section,
a level-1 certificate for~$f_5$ is
\begin{align*}
&      \langle [c_{13}+c_{45}=0],[\x_1\x_2+\x_4\x_6],
           [\coef{1245}+\coef{1346}]\rangle,
\\&      \langle [c_{13}+c_{45}=1],[\x_1\x_3+\x_4\x_5],[\coef{1345}]\rangle,
\end{align*}
and a level-2 certificate for~$f_6$ is
\begin{align*}
&
      \langle [
         c_{34}=0,
         c_{24}+c_{35}=0,
         c_{23}+c_{24}+c_{45}=0
      ], [ \x_1\x_5+\x_2\x_6, \x_1\x_2+\x_1\x_5+\x_5\x_6 ],
\\&\qquad\qquad [
         \coef{1245}+\coef{2356},
         \coef{1235}+\coef{1245}+\coef{2456}
      ]\rangle,
\\&
      \langle [
         c_{34}=1
      ], [ \x_1\x_6, \x_1\x_5+\x_2\x_6 ],
\\&\qquad\qquad [
         \coef{1346},
         \coef{1234}+\coef{1345}+\coef{3456}
      ]\rangle,
\\&
      \langle [
         c_{23}+c_{24}+c_{45}=1
      ], [ \x_1\x_6, \x_1\x_5+\x_2\x_6 ],
\\&\qquad\qquad [
         \coef{1236}+\coef{1246}+\coef{1456},
         \coef{1235}+\coef{1245}+\coef{2456}
      ]\rangle,
\\&
      \langle [
         c_{24}+c_{35}=1,
         c_{23}+c_{24}+c_{45}=0
      ], [ \x_1\x_6, \x_1\x_2+\x_1\x_5+\x_5\x_6 ],
\\&\qquad\qquad [
         \coef{1246}+\coef{1356},
         \coef{1245}+\coef{2356}
      ]\rangle.
\end{align*}
Of course, we have abbreviated the notation by writing, say,
$c_{34}$ instead of $c_{\x_3\x_4}$ and $\coef{1245}$ instead of
$\coef{\x_1\x_2\x_4\x_5}$.

The general result is:

\begin{theorem}
\label{thm:cert}
If $f$ is a Boolean polynomial on $\n$~variables of degree $\n-3$ that has
a level-$k$ certificate, then the minimum weight of $f+\RM(\n-4,\n)$
is at least $2\n+2k+2$.
\end{theorem}

\begin{proof}
Fix a level-$k$ certificate for~$f$.
Let $g$ be an arbitrary member of the coset $f+\RM(\n-4,\n)$; we must show
that $g$ has weight at least~$2\n+2k+2$.  Write $g$ in the form
\begin{equation}
\label{eq:g2}
g = f + \sum_M c_M M + \text{lower-degree terms},
\end{equation}
where $M$ varies over all coquartic monomials.  By certificate
requirement~\ref{item:fcs}, there is a triple
$\langle C,q,r\rangle$ in the certificate such that the
condition~$C$ is satisfied by this~$g$.  Now form a truth table
whose rows are given by $g$, $\x_ig$ for $1\le i\le\n$,
and $q_jg$ for $1 \le j \le k$.  As in
Theorem~\ref{thm:weight} and Theorem~\ref{thm:minwt}, the first
$\n+1$ rows of this table are linearly independent
(using requirement~\ref{item:indepoldold}), self-orthogonal,
and orthogonal to each other.  The last $k$ rows
are self-orthogonal since $q_jg$ has degree less than~$\n$;
these rows are orthogonal to each other by
requirement~\ref{item:orthognewnew} and orthogonal to
the first $\n+1$ rows by requirement~\ref{item:orthogoldnew}.
Requirements \ref{item:indepoldnew} and~\ref{item:indepnewnew}
ensure that each row $q_jg$ is independent of the rows $g$ and
$\x_ig$ and the preceding rows $q_{j'}g$.  Therefore,
the row space of the table has dimension $\n+k+1$ and
codimension at least $\n+k+1$, so the dimension of the ambient
space is at least $2\n+2k+2$.  But all of this goes through
if we delete the entirely-0 columns of the table (where
$g$ is 0), leaving an ambient space whose dimension is
the weight of $g$, so the weight of $g$ is at least $2\n+2k+2$,
as desired.
\end{proof}

\begin{remark}
If a cocubic polynomial $f$ does not satisfy
requirement~\ref{item:indepoldold}, then there
is a nonzero linear polynomial $a = \x_{i_1}+\dots+\x_{i_j}$
such that $af$ has degree less than $\n-2$.  We can apply a linear
transformation of the variables that sends $a$ to $\x_{\n}$;
this transformation maps $f$ to a cocubic polynomial~$f'$ such
that $\x_\n f'$ has degree less than~$\n-2$. It follows that
the degree-$(\n{-}3)$ part of $f'$ is of the form $\x_\n p$
where $p \in \RM(\n-4,\n-1)$.
Polynomials $f$ and~$\x_\n p$ have the same minimum weight,
and one may apply
Proposition~\ref{prop:multvar} to see that the minimum weight of $\x_\n p$ is
equal to the minimum weight of~$p$.
If $p$~still does not satisfy requirement~\ref{item:indepoldold},
we can repeat this process,
eventually obtaining a polynomial $q$ satisfying
requirement~\ref{item:indepoldold} whose minimum weight is equal to the
minimum weight of $f$. A certificate for~$q$ may be used in lieu of a
certificate for $f$.  (Note that, if $f$ and $\x_\n p$ above are homogeneous
cocubics, then $\comppoly{f}$ and $\comppoly{(\x_\n p)}$ are homogeneous
cubics that are linearly equivalent modulo $\RM(2,\n)$, and
$\comppoly{(\x_\n p)}$ does not use variable~$\x_\n$.  Therefore,
if $f$ is a homogeneous cocubic on $\n$ variables such that $\comppoly{f}$ is
known not to be linearly equivalent modulo $\RM(2,\n)$ to a polynomial
on fewer variables, then $f$ must satisfy requirement~\ref{item:indepoldold}.)
\end{remark}

\section{Algorithms for verifying and producing certificates}
\label{sec:certalgs}

Checking the validity of a given certificate entails
verifying six requirements.  All of these verifications
are straightforward except for requirement~\ref{item:fcs}:
the affine subspaces or flats of the subproofs cover
the space of coquartic coefficients.  
(This space is usually far too large to simply exhaust over.)
This is like a validity problem
(the dual of the satisfiability or SAT problem) for a proposition
in conjunctive normal form, so we can borrow a technique from
SAT solvers to solve it.

Namely, we start by making an assumption (an affine equation)
and using it to eliminate a variable from all of the flats;
this yields a reduced problem on one fewer variable, where each
new flat either has one fewer equation than the corresponding
old flat (if the assumption actually followed from the equations
defining the flat), has the same number of equations (if the
assumption is independent of the flat equations), or is
discarded altogether (if the assumption contradicts the flat equations).
We apply the algorithm recursively to this reduced problem.
Then, if that was successful, we apply the algorithm to the
other reduced problem obtained by assuming the negation of the
original assumption.  If that also succeeds, then the verification
is complete.

The efficiency of this recursive
procedure depends heavily on how well each assumption is chosen.
We currently look at the flats with the smallest number
of equations and try to choose an assumption that occurs in
as many of them as possible (preferably both positively and negatively).

What if such a verification fails?  Failure occurs when,
within some number of recursive calls, a state is reached
where we have no flats remaining because all of the flats
have been contradicted by current assumptions.  This means
that the set of current assumptions defines a flat (call it
a ``failure flat'') that is disjoint from all of the flats
in the purported certificate.

While the verification procedure above is somewhat fuzzy, our
procedure for \emph{producing} a level-$k$ certificate is really
more art than science. As with SAT solvers, it is difficult to
settle on one best technique, and it is probably the case that
different strategies work better for different classes of
polynomials $f$. However, we will try to describe our general
strategy below.

First, generate a list of quadratics~$q$ such that
$qf$ has degree at most $\n-2$.  There will be many such
quadratics; the map sending a quadratic~$q$ to the
degree-$(\n{-}1)$ part of $qf$ is a linear map from a space
of dimension $\n(\n-1)/2$ to a space of dimension~$\n$, so its
kernel, the set of suitable quadratics~$q$, has dimension
at least $\n(\n-3)/2$.  In fact, this will be too many to work with
conveniently; it will help to initially restrict to a smaller
set of quadratics, such as a basis of the kernel.  (If we only
had independence requirements, then we could always limit ourselves
to such a basis.  But orthogonality considerations may require
the use of non-basis quadratics; it may be that $q_1g$ and $q_2g$
are not orthogonal, and $q_1g$ and~$q_3g$ are not orthogonal,
but $q_1g$ and $(q_2+q_3)g$ are orthogonal.)

Next, generate new subproofs.  Given a current list of subproofs
(which is initially empty), check whether their flats cover
the coquartic space.  If they do, then we are done; if they don't,
then we get a failure flat, from which we can select a coquartic~$p$.
Consider the particular polynomial $f+p$, and try to find
quadratics $q_1,\dots,q_k$ from our list such that the products
$q_j(f+p)$ are orthogonal to each other and (along with the
products $\x_if$) have linearly independent coquadratic parts.
If we have such quadratics,
then we can select coefficient specifier combinations $r_1,\dots,r_k$
that demonstrate the linear independence.  Now we can generalize from
$f+p$ to $g$, as in~\eqref{eq:g2}, and determine which equations~$C$
on the coefficients~$c_M$ will ensure that certificate requirements
\ref{item:orthognewnew} and~\ref{item:indepnewnew} are satisfied.

There are usually numerous ways to choose $q_1,\dots,q_k$ and
$r_1,\dots,r_k$ as above; we can search through these to find one
for which the success condition~$C$ requires as few equations
as possible, and therefore the corresponding flat covers as much of
the coquartic space as possible.  (This search is one reason why
we prefer to work with a relatively short list of quadratics.)
The result is a triple $\langle C,q,r\rangle$ that we can
add to our list of subproofs; then we start a new iteration,
and continue until the proof is complete or we find a coquartic~$p$
that cannot be handled as above.  (There are various refinements that
can speed this up, such as concentrating on one failure flat and producing
several new subproofs covering parts of it before starting a
new full iteration of the algorithm.)

If we find a coquartic~$p$ for which there is no suitable list
$q_1,\dots,q_k$, then one of three things has occurred:
\begin{itemize}
\item Our current list of quadratics is insufficient; add one or
more quadratics to handle this coquartic and resume the algorithm.
\item The lower bound we are trying to prove isn't true;
see Section~\ref{sec:upper}.
\item The lower bound we are trying to prove is true,
but there is no certificate proving it; see Section~\ref{sec:conclusion}.
\end{itemize}

If the algorithm succeeds in producing a certificate, then this certificate
is likely to have a large number of redundant subproofs, because the
flat for a particular subproof may end up being covered by flats from
earlier and/or later subproofs.  A final pass to remove such
redundant subproofs can result in a considerably shorter certificate.
But this final pass may be quite time-consuming; also, surprisingly,
the shorter certificate may require more time to verify than the
longer certificate, possibly because the heuristic for choosing assumptions
does not perform as well when presented with less information.
(Certificates using fewer quadratics tend to reduce better.)

\section{Algorithms for upper bounds}
\label{sec:upper}

Certifying an upper bound~$w$ on the minimum weight of a cocubic
polynomial~$f$ is easy:
just produce a coquartic~$p$ and verify that the weight of $f+p$
is at most~$w$.  This leaves the problem of finding a good~$p$.

One simple approach is a basic hill climb: start with $f$ (or any other
polynomial in its coset), and try adding monomials of degree at
most~$\n-4$ to it in the hopes of finding a new polynomial of
lower weight than the given one; if successful, start at this new
polynomial and repeat the process.  This method tends to succeed if
the original~$f$ has very few terms, but is less likely to work in
other cases.

This leads to another approach: try to modify $f$ by using linear
transformations of the variables, in the hope of getting another
polynomial~$f'$ such that it is easier to find a low-weight member
of $f'+\RM(\n-4,\n)$.  One can view this as a second hill climb,
this time using linear transformations as the steps (so a basic step
might be swapping two variables or adding one variable to another)
in order to optimize some property of the resulting~$f'$
(which might be the number of distinct variables, the number of terms,
some other property, or a combination of these); of course, one will
probably want to throw away any terms of degree less than $\n-3$
that are produced during this process.
Each polynomial~$f'$ encountered during this second hill climb can be
tried as a starting point for the first hill climb; if a low-weight
polynomial $f'+p' \in f'+\RM(\n-4,\n)$ is found in this way,
and if $f'$ is $f\circ h$ (or the degree-$(\n{-}3)$ part of
$f \circ h$) for some $h \in GL(\n,2)$,
then $(f'+p')\circ h^{-1}$ is a low-weight member of $f+\RM(\n-4,\n)$.

A third approach is to try to prove that the minimum weight of~$f$
is greater than~$w$ and see where the attempt fails.  Let $k=w/2 - \n$,
and apply the algorithm from Section~\ref{sec:certalgs} to try to
produce a level-$k$ certificate for~$f$.  It could be that
$f$ ``almost'' has minimum weight greater than~$w$, meaning that
the flats from the subproofs one can produce cover almost all
of the coquartic space.  The coquartics that are not covered
by these flats are the only possible degree-$(\n{-}4)$ parts
for coquartic polynomials~$p$ such that $f+p$ has weight at most~$w$.
If the number of such coquartics is small enough, then one can try
adding $f$ to each of them and running the first hill climb
(modified so as to alter only terms of degree less than~$\n-4$)
starting from each such sum.  We used this method to improve the
upper bounds on the minimum weights of two of the polynomials
in Section~\ref{sec:78var}.  Note that, even if we fix the
coquartic terms, there will probably be too many polynomials of
degree less than $\n-4$ to run a full exhaust on the remaining
part of the coset $f+\RM(\n-4,\n)$.

\section{Polynomials in 7 or 8 variables}
\label{sec:78var}

Hou~\cite{Hou} provides lists of cubic polynomials on up to
8 variables that include a representative from each equivalence
class of $\RM(3,\n)/\RM(2,\n)$ under the action of
$GL(\n,2)$ for $\n \le 8$; in addition to the polynomials
$f_1,\dots,f_6$ given earlier, there are polynomials
$f_7,\dots,f_{12}$ on 7~variables and $f_{13},\dots,f_{32}$
on 8~variables.  We take the complementary forms to get
a representative list $\comppoly{f_1},\dots,\comppoly{f_{12}}$
of 7-variable cocubic polynomials and a representative list
$\comppoly{f_1},\dots,\comppoly{f_{32}}$
of 8-variable cocubic polynomials.  (Note that $f_1,\dots,f_6$
are all equivalent to their 6-variable complements
under a permutation of variables, so we do not have to redo the work
of Section~\ref{sec:6var} in complementary form.)

By Proposition~\ref{prop:multvar}, we only need to handle the
polynomials $\comppoly{f_7},\dots,\comppoly{f_{12}}$ for
7~variables, and then the polynomials
$\comppoly{f_{13}},\dots,\comppoly{f_{32}}$ for 8~variables.

\begin{table}[ht]
\begin{center}
\begin{tabular}{|c|c|c|c|c|c|}
   \hline
   Polynomial & Weight & Min.\ weight & Closest in $RM(3,7)$ & Cert.\ level &
        Subproofs
   \\ \hline
   $\comppoly{f_7}$ & 20 & 16 & $\x_1\x_2\x_3$ & 0 & 1
   \\ $\comppoly{f_8}$ & 18 & 16 & $\x_1\x_2\x_3$ & 0 & 1
   \\ $\comppoly{f_9}$ & 20 & 20 & 0 & 2 & 12
   \\ $\comppoly{f_{10}}$ & 18 & 18 & 0 & 1 & 2
   \\ $\comppoly{f_{11}}$ & 22 & 20 &
         $\begin{aligned} &\x_1\x_2\x_7+\x_1\x_5\x_7
            \\[-5pt]&\quad{}+\x_1\x_6\x_7+\x_4\x_6\x_7
         \end{aligned}$ & 2 & 29
   \\ $\comppoly{f_{12}}$ & 24 & 20 & $\x_3\x_4\x_5$ & 2 & 13
   \\ \hline
\end{tabular}
\caption{Minimum weights of cocubics on 7 variables}
\label{7vartable}
\end{center}
\end{table}

Table~\ref{7vartable} shows the results obtained for
polynomials $\comppoly{f_7},\dots,\comppoly{f_{12}}$.
It shows the weight and minimum weight for each polynomial.
One can verify that the specified value is an upper bound for
the minimum weight by checking that the weight of $\comppoly{f_i}+p_i$
is the specified value, where $p_i$ is the polynomial in the
``Closest in $RM(3,7)$'' column.  To get lower bounds,
we produced certificates of the levels given in the ``Cert.\ level''
column; the last column gives the number of subproofs in
each certificate (after full simplification).

These computations give a new (and quickly verifiable) proof of the 
result of Wang, Tan, and Prabowo \cite{Wang-Tan-Prabowo}
that the covering radius of $\RM(3,7)$ in $\RM(4,7)$
is~20.

\begin{table}[ht]
\begin{center}
\begin{tabular}{|c|c|c|c|c|}
   \hline
   Polynomial & Weight & Min.\ weight & Cert.\ level & Subproofs
   \\ \hline
   $\comppoly{f_{13}}$ & 20 & 18 & 0 & 1
   \\ $\comppoly{f_{14}}$ & 18 & 18 & 0 & 1
   \\ $\comppoly{f_{15}}$ & 20 & 20 & 1 & 2
   \\ $\comppoly{f_{16}}$ & 22 & 20 & 1 & 2
   \\ $\comppoly{f_{17}}$ & 24 & 22 & 2 & 4
   \\ $\comppoly{f_{18}}$ & 28 & 22 & 2 & 35
   \\ $\comppoly{f_{19}}$ & 26 & 22 & 2 & 33
   \\ $\comppoly{f_{20}}$ & 30 & 22 & 2 & 20
   \\ $\comppoly{f_{21}}$ & 28 & 22 & 2 & 41
   \\ $\comppoly{f_{22}}$ & 28 & 22 & 2 & 81
   \\ $\comppoly{f_{23}}$ & 26 & 22 & 2 & 29
   \\ $\comppoly{f_{24}}$ & 28 & 22 & 2 & 52
   \\ $\comppoly{f_{25}}$ & 28 & 20 & 1 & 4
   \\ $\comppoly{f_{26}}$ & 28 & 24 & 3 & 478
   \\ $\comppoly{f_{27}}$ & 32 & 26 & 4 & 10022
   \\ $\comppoly{f_{28}}$ & 24 & 22 & 2 & 29
   \\ $\comppoly{f_{29}}$ & 20 & 20 & 1 & 4
   \\ $\comppoly{f_{30}}$ & 22 & 20 & 1 & 2
   \\ $\comppoly{f_{31}}$ & 22 & 20 & 1 & 2
   \\ $\comppoly{f_{32}}$ & 28 & 24 & 3 & 179
   \\ \hline
\end{tabular}
\caption{Minimum weights of cocubics on 8 variables}
\label{8vartable}
\end{center}
\end{table}

Table~\ref{8vartable} shows the results obtained for
polynomials $\comppoly{f_{13}},\dots,\comppoly{f_{32}}$.
The columns are the same as for Table~\ref{7vartable}, except that
the ``Closest in $\RM(4,8)$'' column has been omitted because
several of these polynomials are too large to fit.
(It may be of interest to note that, for $\comppoly{f_{20}}$
and~$\comppoly{f_{25}}$, the closest polynomial in~$\RM(4,8)$
cannot be a homogeneous coquartic; lower-degree terms are
required.  We verified this by using failed attempts to get
a level-3 certificate for~$\comppoly{f_{20}}$ and a level-2
certificate for~$\comppoly{f_{25}}$; the failures produced
explicit lists of what homogeneous coquartics could possibly
be added to these polynomials to get the desired weight,
and we checked that none of the coquartics on these lists
actually worked.)

The initial level-4 certificate for $\comppoly{f_{27}}$ took several weeks
on 90 processors to produce (although program improvements were being made
simultaneously, so repeating the computation would probably take less time).
This initial certificate had 64534 subproofs; 27~different quadratics
were used.  Later processing
simplified the certificate to 10022 subproofs.
Verifying this certificate in Magma took about 4.7~hours (almost all of which
is spent verifying certificate requirement~\ref{item:fcs}; a separate
C~program verified this one requirement in 33~minutes, while a CUDA version
running on a GPU took 4~minutes).

These computations show that the covering radius of $\RM(4,8)$ in $\RM(5,8)$
is~26.

\section{Polynomials in 9 variables}
\label{sec:9var}

We used the Brier--Langevin~\cite{Brier-Langevin-web} list
of 349 representative 9-variable cubic polynomials (actually, we produced
simplified versions of them via linear
transformations, which we will call $\BLpoly{i}$ for $1 \le i \le 349$),
and took their complementary forms to give representative
cocubic (degree-6) polynomials on 9~variables.

All have minimum weight at most 32, and all but six
(numbers 107, 148, 165, 274, 301, and~329)
have minimum weight at most 30.
(We have coquartics that can be
added to these polynomials to verify this.)

We have produced a level-4 certificate for the complementary form
$\comppoly{\BLpoly{311}}$ of the polynomial
\begin{equation*}
\BLpoly{311} =
\x_1\x_2\x_3 + \x_1\x_4\x_6 + \x_1\x_7\x_9 + \x_2\x_4\x_9 + 
    \x_2\x_5\x_6 + \x_3\x_5\x_9 + \x_6\x_8\x_9.
\end{equation*}
The original certificate had 21697 subproofs
using 33 quadratics; further computation simplified this
to 6134 subproofs.  So
the minimum weight of~$\comppoly{\BLpoly{311}}$
is at least~28; in fact, we have a coquartic that demonstrates
that this minimum weight is exactly~28.

So the covering radius of $\RM(5,9)$ in $\RM(6,9)$ is at least~$28$
and at most~$32$.  (We conjecture that the actual value is $32$;
we are just unable to produce a level-6 certificate at present.)

\section{Polynomials in more variables}
\label{sec:morevar}

We have also produced a level-4 certificate for the complementary form
of the ten-variable polynomial
\begin{equation}
\label{eq:TV}
\x_1\x_3\x_5 + \x_1\x_4\x_{10} + \x_2\x_3\x_6 + \x_2\x_4\x_8 + 
    \x_2\x_5\x_{10} + \x_3\x_4\x_7 + \x_8\x_9\x_{10}.
\end{equation}
The certificate has 65059 subproofs
using 37 quadratics (before simplification, which is in progress).  So
the minimum weight of the complement of~\eqref{eq:TV}
is at least~30; again, we have a coquartic that demonstrates
that this minimum weight is exactly~30.

We now show how to use the examples already produced
to get similar examples on higher numbers of variables.

\begin{theorem}
\label{thm:upward}
If the cocubic Boolean polynomial~$f$ on variables~$\x_1,\dots,\x_{\n}$
has a level-$k$ certificate and the cocubic polynomial~$f'$
on new variables $\y_1,\dots,\y_{\n'}$ has a level-0 certificate,
then the cocubic polynomial $f^*=\y_1\cdots\y_{\n'}f+\x_1\cdots\x_{\n}f'$
has a level-$k$ certificate.
\end{theorem}

The proof of Theorem~\ref{thm:upward} proceeds by constructing a
new certificate for~$f^*$ and verifying that requirements
\ref{item:indepoldold}--\ref{item:fcs} hold for it; the details
are given in the Appendix.

Given an $\n$-variable polynomial~$f(\x_1,\dots,\x_{\n})$
of degree $\n-3$ with a level-$k$ certificate,
we can apply Theorem~\ref{thm:upward} with
$f'=1$ on variables $\y_1,\y_2,\y_3$ to get an $(\n+3)$-variable
polynomial~$f^*$ with a level-$k$ certificate.
Since we already have 8-, 9-, and 10-variable polynomials
with level-4 certificates, we can now generate $\n$-variable
polynomials with level-4 certificates for all $\n \ge 8$.
Therefore, Theorem~\ref{thm:cert} gives:

\begin{theorem}
The covering radius of $\RM(\n-4,\n)$
in $\RM(\n-3,\n)$ is at least $2\n+10$ for all $\n \ge 8$.
\end{theorem}

This improves a lower bound result of McLoughlin~\cite{McLoughlin},
stating that the covering radius of $\RM(\n-4,\n)$ in $\RM(\n-3,\n)$
is at least $2\n+2$ for odd~$\n$ and $2\n$ for even~$\n$.
It is not as good asymptotically as the sphere-covering
bound~\cite{MacWilliams-Sloane}, which gives a quadratic lower bound
on this covering radius.  The most basic way to apply the
sphere-covering bound gives that
the covering radius of $\RM(\n-4,\n)$ in $\RM(\n-3,\n)$ is bounded below
by the least~$k$ such that
\begin{equation}
\label{scb1}
\sum_{i=0}^k \binom{2^\n}{i} \ge
2^{\raisebox{1ex}{\mbox{$\textstyle \binom{\n}{3}$}}}.
\end{equation}
One can improve this using the fact that the code $\RM(\n-3,\n)$
has minimum distance 8, and hence the balls of radius~3 centered at
the codewords are disjoint; this yields as a lower bound
the least $k$ such that
\begin{equation}
\label{scb2}
\sum_{i=0}^{k+3} \binom{2^\n}{i} \ge
\left( 1 + 2^\n + \binom{2^\n}{2} + \binom{2^\n}{3} \right)
2^{\raisebox{1ex}{\mbox{$\textstyle \binom{\n}{3}$}}}.
\end{equation}
These bounds exceed $2\n+10$ for $\n \ge 15$.  (However, the sphere-covering
bound is nonconstructive and does not yield explicit examples of
cosets of high minimum weight.)

\begin{remark}
Inequality~\eqref{scb2} implies inequality~\eqref{scb1} because
\begin{equation*}
\sum_{i=0}^{k+3} \binom{2^\n}{i} \le
   \left(\sum_{i=0}^3 \binom{2^\n}{i}\right)
   \left(\sum_{i=0}^k \binom{2^\n}{i}\right),
\end{equation*}
which is true because any length-$2^\n$ bit word of Hamming weight
at most $k+3$ can be expressed in at least one way (usually many) as the sum
of a word of Hamming weight at most~3 and a word of Hamming weight
at most~$k$.
\end{remark}

It would be nice to get a better result combining a level-$k$ certificate
for~$f$ and a level-$k'$ certificate for~$f'$ to get a
level-$(k{+}k')$ certificate for~$f^*$, but the straightforward
way of doing that (by combining each triple of the first certificate
with each triple of the second certificate) does not yield
independence or orthogonality
between quadratics from the first certificate and quadratics from
the second certificate
(certificate requirements \ref{item:orthognewnew} and~\ref{item:indepnewnew}).

However, one can get another partial result:

\begin{theorem}
\label{thm:1+1=2}
If $f$, $f'$, and $f^*$ are as in Theorem~\ref{thm:upward} and
$f$ and~$f'$ both have
level-1 certificates, then $f^*$ has a level-2 certificate.
\end{theorem}

Again the proof proceeds by using the two given certificates to construct a
new certificate for~$f^*$ and verifying that requirements
\ref{item:indepoldold}--\ref{item:fcs} hold for it; the details
are given in the Appendix.

\section{Cryptographic properties and testing polynomials}
\label{sec:crypto}

A number of properties of Boolean polynomials have been identified
as relevant for measuring the strength of these polynomials in
cryptographic applications; these properties include balance,
resiliency, high nonlinearity, and
algebraic immunity.  This topic is far too broad to even summarize
here; overviews of it are given in
\cite{Carlet-chapter}, \cite{CS-book}, and~\cite{Mesnager-book}.
But here is a brief review of a few of the definitions.
An $\n$-variable Boolean function (polynomial)~$f$ is
\emph{balanced} if it has weight $2^{\n-1}$.  Polynomial~$f$ is
\emph{\ordinal{$t$}-order correlation immune} if, for any $k$ with
$1 \le k \le t$ and any distinct variables $\x_{i_1},\dots,\x_{i_k}$,
$f+\x_{i_1}+\dots+\x_{i_k}$ is balanced.  A polynomial that is both
balanced and \ordinal{$t$}-order correlation immune is called
\emph{$t$-resilient}.  The \emph{nonlinearity} of~$f$ is the
distance from~$f$ to $\RM(1,\n)$.

One may wonder how the minimum weight of a cocubic polynomial interacts with
these cryptographic properties. There are two main questions one might ask:
\begin{enumerate}
\item Do the representative cocubic polynomials or their minimum-weight
  coset members from the previous sections exhibit any of these
  cryptographic properties?
\item What are the minimum weights of cocubic polynomials that do exhibit
  those cryptographic properties?
\end{enumerate}
We will address these questions in turn, starting with Question 1.

The representative cubic polynomials on the lists from
\cite{Hou} and~\cite{Brier-Langevin-web} were chosen at least partially
on the basis of simplicity
(lower numbers of terms), which means that the
corresponding cocubic polynomials have very low weight,
even before we modify them to get minimum-weight members of
their cosets.  It then follows from known results (see the
references above) that the polynomials we work with here
cannot have most of the cryptographic properties listed above.
The exception is \ordinal{$t$}-order correlation immunity,
which is possible even for highly unbalanced polynomials.
We checked and found that
none of the representative cocubics or their minimum-weight versions
is first-order correlation immune, except for the zero polynomial.

These representative polynomials seem unlikely to be of
cryptographic interest directly, so we now turn our attention to
Question 2. Given a cryptographically-interesting polynomial, we
can test its minimum weight
by using invariants
from \cite{Hou} and \cite{Brier-Langevin} to determine which
representative polynomial it is equivalent to.
Here and for the rest of this section, 
``equivalent'' means
``linearly equivalent modulo coquartic polynomials'' or
``linearly equivalent over $\RM(\n-4,\n)$.''

\begin{remark}
In \cite{Brier-Langevin}, the direct product of two
new invariants is used to discriminate the 349 equivalence classes
of 9-variable cubic polynomials, but we found
that the second new invariant alone suffices.
\end{remark}

\begin{table}[ht]
\begin{center}
\begin{tabular}{|c|c|c|c|c|}
   \hline
   Parameters & Source & Equivalent to & Min.~weight
   \\ \hline
   (7,0,4,44) & Symmetric & $\comppoly{f_7}$ & 16
   \\ {}[8,0,5,64] & Symmetric & $\comppoly{f_7}$ & 16
   \\ {}[9,0,6,120] & Symmetric & $\comppoly{\BLpoly{346}}$ & 20
   \\ (8,2,5,112) & \cite[Th.~10(d)]{Sarkar-Maitra-2000}
      & $\comppoly{f_2}$ & 8
   \\ {}[8,1,5,116] & \cite[p.~1828]{Maitra-Pasalic}    % 2002
      & $\comppoly{f_2}$ & 8
   \\ (9,2,6,232) & \cite[Pr.~4(1)]{Sarkar-Maitra-2004}
      & $\comppoly{\BLpoly{348}} = \comppoly{f_2}$ & 8
   \\ (9,2,6,240) & \cite[App.~B]{Clark}      % 2004
      & $\comppoly{\BLpoly{321}} = \comppoly{f_{14}}$ & 18
   \\ (6,2,3,24) & \cite[Lemma~3.3]{Borissov}      %2005
      & $f_5$ & 16
   \\ (8,2,5,16) & \cite[Th.~5.8]{Borissov}      %2005
      & $\comppoly{f_7}$ & 16
   \\ (7,1,4,40) & \cite[Pr.~5.13]{Borissov}      %2005
      & $\comppoly{f_9}$ & 20
   \\ (7,2,4,56) & \cite[Table~4]{Stanica-Maitra}     % 2008
      & $\comppoly{f_9}$ & 20 
   \\ (10,2,7,488) & \cite[Table~1]{Liu-Youssef}     % 2009
      &   & 28--38
   \\ \hline
\end{tabular}
\caption{Minimum weights of various known cocubic polynomials}
\label{tab:cpweights}
\end{center}
\end{table}

Table~\ref{tab:cpweights} shows the results we got by testing
various cocubic polynomials from the literature.  In each case
(except the last), we found out which representative cocubic
was equivalent to the given one, and looked up the corresponding
minimum weight.

The ``Parameters'' column uses a notation appearing in
\cite{Maitra-Pasalic} and other sources.
An $(n,m,d,x)$ polynomial (or function) is a Boolean polynomial on
$n$~variables of degree~$d$ which is $m$-resilient and has nonlinearity~$x$.
The unbalanced version of this is:
An $[n,m,d,x]$ polynomial (or function) is a Boolean polynomial on
$n$~variables of degree~$d$ which is \ordinal{$m$}-order correlation immune
and has nonlinearity~$x$.

Other notes on Table~\ref{tab:cpweights}:
\begin{itemize}
\item
``Symmetric'' means the symmetric homogeneous cocubic polynomial
on $\n$~variables.  These polynomials were used
in \cite{McLoughlin} to get the lower bound mentioned in
Section~\ref{sec:morevar}.
\item The equivalences and minimum weights for the polynomials
from~\cite{Borissov} were given in \cite{Borissov}.
% error in Borissov p. 1187 table: last two vectors should be
%   (1,1,0,0,0,0,...,0,0,0,0)
%   (0,0,1,0,1,0,...,1,0,1,0)
\item
All 72 of the rotationally symmetric (7,2,4,56)
polynomials in \cite[Table~4]{Stanica-Maitra}
are in the same equivalence class.
\item
For the two 10-variable polynomials from~\cite{Liu-Youssef},
since we do not have a precomputed catalog for 10~variables,
we applied our methods directly to get lower and upper bounds
on the minimum weight; the same bounds were obtained for both.
(We verified using the invariants in~\cite{Hou}
that these two polynomials are not equivalent.)
\end{itemize}

For convenience, we show the relevant representative cubic polynomials
(other than those in Section~\ref{sec:6var}) here:
\begin{align*}
f_7 &= \x_1\x_2\x_7 + \x_3\x_4\x_7 + \x_5\x_6\x_7,
\\ f_9 &= \x_1\x_2\x_3 + \x_1\x_4\x_7 + \x_2\x_4\x_5 + \x_3\x_4\x_6,
\\ f_{14} &= \x_1\x_2\x_3 + \x_1\x_7\x_8 + \x_4\x_5\x_6 + \x_4\x_7\x_8,
\\ \BLpoly{346} &= \x_1\x_2\x_9 + \x_3\x_4\x_9 + \x_5\x_6\x_9 + \x_7\x_8\x_9.
\end{align*}

A new type of covering radius was introduced in~\cite{Kurosawa}:
$\hat\rho(t,r,n)$ is defined to be the maximum distance between
a $t$-resilient function and $\RM(r,n)$.  Since we are working with
$\RM(\n-4,\n)$ in this paper, and since the Siegenthaler
bound~\cite[Th.~1]{Siegenthaler} states that a $t$-resilient function
on $\n$ variables has degree at most $\n-t-1$, the instances
of~$\hat\rho$ that are most relevant here are of the form
$\hat\rho(2,\n-4,\n)$.

It is known that $\hat\rho(2,1,5)=8$~\cite{Kurosawa}
and $\hat\rho(2,2,6)=16$~\cite{Borissov}. Reference~\cite{Borissov} also
gives the bounds $16\le\hat\rho(2,3,7)\le22$ and
$16\le\hat\rho(2,4,8)$, while
\cite{Wang-Tan-Prabowo}~gives $\hat\rho(2,3,7)\le 20$.

The Siegenthaler bound implies that
$\hat\rho(2,\n-4,\n)$ is at most the covering radius of
$\RM(\n-4,\n)$ in $\RM(\n-3,\n)$; therefore, 
our upper bound results imply $\hat\rho(2,4,8)\le 26$
and $\hat\rho(2,5,9)\le 32$.
For $\hat\rho(2,6,10)$ we do not know of an upper bound
better than~50, obtained from the facts that the covering
radius of~$\RM(6,10)$ is at most~51
\cite[Table~9.1 and (9.3.4)]{Cohen-book}
and that $\hat\rho(2,6,10)$ must be even because all polynomials
involved have degree less than~$10$.

Our computations in Table~\ref{tab:cpweights} give three improvements
to the known lower bounds on $\hat\rho(2,\n-4,\n)$: the 
(7,2,4,56) polynomials from~\cite{Stanica-Maitra}
give $\hat\rho(2,3,7) \ge 20$, the (9,2,6,240)
polynomial from~\cite{Clark} gives
$\hat\rho(2,5,9) \ge 18$, and the (10,2,7,488)
polynomials from~\cite{Liu-Youssef} give
$\hat\rho(2,6,10) \ge 28$.

So the updated bounds on $\hat\rho(2,\n-4,\n)$ for $7 \le \n \le 10$ are:
\begin{gather*}
\hat\rho(2,3,7) = 20,
\\ 16 \le \hat\rho(2,4,8) \le 26,
\\ 18 \le \hat\rho(2,5,9) \le 32,
\\ 28 \le \hat\rho(2,6,10) \le 50.
\end{gather*}

\section{Conclusion and open questions}
\label{sec:conclusion}

We have given a method for producing verifiable certificates
of lower bounds for the minimum weights of cosets of $\RM(\n-4,\n)$
in $\RM(\n-3,\n)$, and described a simple method for searching
for upper bounds on these minimum weights.  Using these methods,
we have improved the known bounds on the covering radius of
$\RM(\n-4,\n)$ in $\RM(\n-3,\n)$ for $8 \le \n \le 14$, with an exact
value of~$26$ for $\n=8$.  We have also improved the known bounds
on the modified covering radius $\hat\rho(2,\n-4,\n)$ for
$7 \le \n \le 10$, with an exact value of~$20$ for $\n=7$.

It would be a massive project to apply these methods
to compute the covering radius of $\RM(6,10)$ in $\RM(7,10)$;
the list of representative polynomials
would have length 3691561~\cite{Hou}, and we would
need a good upper bound for each, although we only need
a lower bound for one of them.
For $m>10$, the situation is even worse: the number of cosets
of $\RM(\n-4,\n)$ in $\RM(\n-3,\n)$ is
$2^{\raisebox{1ex}{\mbox{$\scriptstyle \binom{m}{3}$}}}$,
while the size of $GL(\n,2)$ is less than $2^{\n^2}$, 
so the number of representatives needed is greater than
$2^{\raisebox{1ex}{\mbox{$\scriptstyle \binom{m}{3}-m^2$}}}
= 2^{(m^2(m-9)+2m)/6}$.

Do there exist certificates of arbitrarily high level?
In particular, can one combine a level-$k$ certificate
with a level-$k'$ certificate to form a level-$(k{+}k')$
certificate as in Section~\ref{sec:morevar}, even
for a restricted family of polynomials $f$ and~$f'$?

Do the methods extend to handle \ordinal{$(\n-4)$}-order nonlinearity
for $\n$-variable polynomials of degree greater than $\n-3$?
In other words, can they be used to compute the covering
radius of $\RM(\n-4,\n)$ within the entire space $\RM(\n,\n)$?
This is not clear at present; the basic lower bound results
would need substantial revision.

Are these methods sufficient in general?  Given a cocubic polynomial~$f$
on $\n$ variables at distance~$d$ from $\RM(\n-4,\n)$, will there always
exist a certificate for~$f$ giving a lower bound of~$d$?  Might one need
to look for other truth-table rows besides those of the form~$qf$
with $q$~quadratic?

Are the problems studied here provably difficult?  It is known that
the problem of verifying an upper bound on the covering radius
of a general binary linear code is complete for the second
level of the polynomial-time hierarchy~\cite{McLoughlin2}.
Can one prove a lower bound on
complexity that is specific to Reed--Muller codes?  For instance,
is the problem of verifying a given upper bound on the
distance from a given polynomial to $\RM(r,\n)$ NP-complete?

\appendix
\section{Appendix: Proofs of Theorems \ref{thm:upward} and \ref{thm:1+1=2}}

\begin{proof}[Proof of Theorem \ref{thm:upward}]
We get a level-$k$ certificate for~$f^*$ by changing each triple
$\langle C,q,r \rangle$ of the level-$k$ certificate for~$f$
into a new triple $\langle C^*,q^*,r^* \rangle$ as follows:
\begin{itemize}
\item Each variable $c_M$ occurring in~$C$ (where $M$ is a coquartic monomial
in the variables $\x_1,\dots,\x_{\n}$) is replaced in $C^*$ with
$c_{M^*}$, where $M^* = M\y_1\cdots\y_{\n'}$.
\item $q^* = q$.
\item Each coefficient specifier $\coef{M}$ occurring in~$r$
(where $M$ is a coquadratic monomial in the variables $\x_1,\dots,\x_{\n}$)
is replaced in $C^*$ with $\coef{M^*}$, where $M^* = M\y_1\cdots\y_{\n'}$.
\end{itemize}
We now verify that all of the certificate requirements are
met by the new certificate.

Requirement~\ref{item:indepoldold}: The degree-$(\n{+}\n'{-}2)$ parts
of the polynomials $\x_if^*$ are just $\y_1\cdots\y_{\n'}$ times the
degree-$(\n{-}2)$ parts of the polynomials $\x_if$, since $\x_i$
times $\x_1\cdots\x_{\n}f'$ is just $\x_1\cdots\x_{\n}f'$, which has
degree $\n+\n'-3$.  Similarly, the degree-$(\n{+}\n'{-}2)$ parts
of the polynomials $\y_if^*$ are just $\x_1\cdots\x_{\n}$ times the
degree-$(\n'{-}2)$ parts of the polynomials $\y_if'$.  These two
collections of polynomials are separately linearly independent,
and they are supported on disjoint sets of monomials, so they
are jointly linearly independent.

Requirement~\ref{item:orthogoldnew}: We have
$q_jf^* = \y_1\cdots\y_{\n'}q_jf+\x_1\cdots\x_{\n}q_jf'$,
where $q_jf$ has degree at most $\n-2$ by assumption and
$\x_1\cdots\x_{\n}q_jf'$ is either $\x_1\cdots\x_{\n}f'$ or 0,
depending on the number of terms in~$q_j$, so
$q_jf^*$ has degree at most $\n+\n'-2$.

Requirement~\ref{item:orthognewnew}: Any sum $\sum_{M^*} c_{M^*}M^*$
of coquartic terms can be split into two parts: the part $S_1$
which is the sum of those terms for which $M^*$ has the
form $M\y_1\cdots\y_{\n'}$
for some coquartic monomial $M$ on $\x_1,\dots,\x_{\n}$, and
the part $S_2$ which is the sum of those terms for which $M^*$~is not a
multiple of $\y_1\cdots\y_{\n'}$.  The product
$q_{j'}q_{j}(\y_1\cdots\y_{\n'}f+S_1)$ has degree less than $\n+\n'$
by requirement~\ref{item:orthognewnew} for the given certificate, while
$q_{j'}q_{j}S_2$ has no terms that are multiples of $\y_1\cdots\y_{\n'}$,
and $q_{j'}q_{j}\x_1\cdots\x_{\n}f'$ is again either $\x_1\cdots\x_{\n}f'$
or~0.  Therefore, $q_{j'}q_{j}(f^*+\sum_{M^*} c_{M^*}M^*)$ has degree
less than $\n+\n'$.

Requirement~\ref{item:indepoldnew}: Since every monomial referred to
by the coefficient combination $r^*_j$ has $\y_1\cdots\y_{\n'}$ as
a factor, $r^*_j$ for $\x_i\x_1\cdots\x_{\n}f'$ or $\y_i\x_1\cdots\x_{\n}f'$
is 0.  The combination $r^*_j$ for $\y_i\y_1\cdots\y_{\n'}f$ is 0
since $\y_i\y_1\cdots\y_{\n'}f = \y_1\cdots\y_{\n'}f$ has degree less than
$\n+\n'-2$, while $r^*_j$ for $\x_i\y_1\cdots\y_{\n'}f$ is 0
by requirement~\ref{item:indepoldnew} for the given certificate.
Therefore, $r^*_j$ for $\x_if^*$ or for~$\y_if^*$ is 0.

Requirement~\ref{item:indepnewnew}: Let us again split the sum
$\sum_{M^*} c_{M^*}M^*$
of coquartic terms into two parts $S_1$ and~$S_2$ as above.
Since every monomial referred to
by $r^*_j$ has $\y_1\cdots\y_{\n'}$ as
a factor, $r^*_j$ for $q_{j'}\x_1\cdots\x_{\n}f'$ is 0, and
so is $r^*_j$ for $q_{j'}S_2$.  And $r^*_j$ for
$q_{j'}(\y_1\cdots\y_{\n'}f+S_1)$ is 0 if $j'<j$ or 1~if $j'=j$
by requirement~\ref{item:indepnewnew} for the given certificate.
Therefore, $r^*_j$ for $q_{j'}(f^*+\sum_{M^*} c_{M^*}M^*)$ is 0 if
$j'<j$, 1 if $j'=j$, as required.

Requirement~\ref{item:fcs}: The conditions~$C^*$ are the same as
the conditions~$C$ except for a global renaming of variables; since
every assignment to the old variables satisfies at least one condition~$C$,
every assignment to the new variables satisfies at least one condition~$C^*$.
(The fact that there are additional new variables not mentioned by
any~$C^*$ at all does not affect this.)

So the new certificate is valid, and we are done.
\end{proof}

\begin{proof}[Proof of Theorem \ref{thm:1+1=2}]
We construct the new certificate as follows.  First, for each
subproof $\langle C,q,r \rangle$ in the certificate for~$f$ and each
subproof $\langle C',q',r' \rangle$ in the certificate for~$f'$,
create a subproof $\langle C^*,q^*,r^* \rangle$ where:
\begin{itemize}
\item $C^*$ consists of the equations in~$C$ with all subscripts
multiplied by $\y_1\cdots\y_{\n'}$, the equations in~$C'$ with all
subscripts multiplied by $\x_1\cdots\x_{\n}$, and two new equations
described below.
\item $q^* = [q_1,q'_1]$.
\item $r^* = [r^*_1,r^*_2]$ where $r^*_1$ is $r_1$ with all
monomials multiplied by $\y_1\cdots\y_{\n'}$ and $r^*_2$ is $r'_1$ with all
monomials multiplied by $\x_1\cdots\x_{\n}$.
\end{itemize}
The two new equations are of the form
\begin{equation}
\label{eq:neweq}
\sum_{i,j,i',j'} c_{\comppoly{(\x_i\x_j\y_{i'}\y_{j'})}}=0;
\end{equation}
in the first equation the sum runs over all $i<j$ and $i'<j'$
where $\x_i\x_j$ is a monomial of~$q_1$ and $\y_{i'}\y_{j'}$
is a monomial of~$q'_1$, while
in the second equation the sum runs over all $i<j$ and $i'<j'$
where $\x_i\x_j$ is a monomial of~$q_1$ and $\comppoly{(\y_{i'}\y_{j'})}$
is a monomial mentioned in~$r'_1$.  These two equations are
needed to ensure that certificate requirements
\ref{item:orthognewnew} and~\ref{item:indepnewnew} hold for the
new subproof.  But now, in order to get certificate requirement~\ref{item:fcs}
to hold for the new certificate, we need some additional subproofs:
for all $(i,j,i',j')$ such that $c_{\comppoly{(\x_i\x_j\y_{i'}\y_{j'})}}$
is mentioned in one of the new equations~$\eqref{eq:neweq}$,
add the new triple
\begin{equation}
\label{eq:newsubproof}
\langle [c_{\comppoly{(\x_i\x_j\y_{i'}\y_{j'})}}=1],
[\x_i\y_{i'},\x_i\y_{j'}],
[\coef{\comppoly{(\x_j\y_{j'})}},\coef{\comppoly{(\x_j\y_{i'})}}] \rangle
\end{equation}
to the new certificate.
We will now see that combining the subproofs $\langle C^*,q^*,r^* \rangle$
above with the new triples~\eqref{eq:newsubproof}
gives a complete level-2 certificate for~$f^*$.

Requirement~\ref{item:indepoldold}: The same proof as for
Theorem~\ref{thm:upward} works here.

Requirement~\ref{item:orthogoldnew}: For
$\langle C^*,q^*,r^* \rangle$, the proof for Theorem~\ref{thm:upward}
works for~$q_1$, and the same with the roles of $\x$ and~$\y$ reversed
works for~$q'_1$.  For~\eqref{eq:newsubproof}, just note that
$\x_i\y_{i'}f^* = \x_i\y_1\cdots\y_{\n'}f+\y_{i'}\x_1\cdots\x_{\n}f'$.

Requirement~\ref{item:orthognewnew}: For $\langle C^*,q^*,r^* \rangle$,
note that $q_1q'_1(f^*+\sum_{M^*} c_{M^*}M^*)$ is the sum of parts
$q_1q_1'\y_1\cdots\y_{\n'}f$, $q_1q_1'\x_1\cdots\x_{\n}f'$, and
$q_1q_1'\sum_{M^*} c_{M^*}M^*$; the first part has no terms that are
multiples of $\x_1\cdots\x_{\n}$, the second part has no terms that are
multiples of $\y_1\cdots\y_{\n'}$, and, in the third part, the coefficient of
$\x_1\cdots\x_{\n}\y_1\cdots\y_{\n'}$ is the left-hand side of
the first equation~\eqref{eq:neweq}, which is assumed to be~$0$.
For~\eqref{eq:newsubproof} the argument is the same, except that the
third part $\x_i\y_{i'}\x_i\y_{j'}\sum_{M^*} c_{M^*}M^*$ has no
terms of degree $\n+\n'$ because $\x_i\y_{i'}\x_i\y_{j'}=\x_i\y_{i'}\y_{j'}$
has degree less than~4.

Requirement~\ref{item:indepoldnew}: For
$\langle C^*,q^*,r^* \rangle$, the proof for Theorem~\ref{thm:upward}
works for~$r^*_1$, and the same with the roles of $\x$ and~$\y$ reversed
works for~$r^*_2$.  For~\eqref{eq:newsubproof}, note that every
monomial in~$f^*$ is missing either at least three~$\x$'s or at
least three~$\y$'s, so every monomial in $\x_\ell f^*$ or $\y_\ell f^*$
is missing either at least two~$\x$'s or at least two~$\y$'s.

Requirement~\ref{item:indepnewnew}: For
$\langle C^*,q^*,r^* \rangle$, the proof that
$r^*_1$ for $q_1(f^*+\sum_{M^*} c_{M^*}M^*)$ is~1
is the same as for Theorem~\ref{thm:upward}, and the proof that
$r^*_2$ for $q'_1(f^*+\sum_{M^*} c_{M^*}M^*)$ is~1
is the same with the roles of $\x$ and~$\y$ reversed.
To see that $r^*_2$ for $q_1(f^*+\sum_{M^*} c_{M^*}M^*)$ is~0,
separate this product into parts $q_1\y_1\cdots\y_{\n'}f$,
$q_1\x_1\cdots\x_{\n}f'$, and $q_1\sum_{M^*} c_{M^*}M^*$;
the first part has no terms that are multiples of $\x_1\cdots\x_{\n}$,
the second part has degree at most~$\n+\n'-3$, and
$r^*_2$ for the third part is the left-hand side of
the second equation~\eqref{eq:neweq}, which is assumed to be~$0$.
For~\eqref{eq:newsubproof}, note that all terms of
$\x_i\y_{i'}f^*$ or $\x_i\y_{j'}f^*$ are multiples of either
$\x_1\cdots\x_{\n}$ or $\y_1\cdots\y_{\n'}$, so we compute that
$\coef{\comppoly{(\x_j\y_{j'})}}$ for $\x_i\y_{i'}(f^*+\sum_{M^*} c_{M^*}M^*)$
is $c_{\comppoly{(\x_i\x_j\y_{i'}\y_{j'})}}$,
$\coef{\comppoly{(\x_j\y_{i'})}}$ for $\x_i\y_{i'}(f^*+\sum_{M^*} c_{M^*}M^*)$
is~0, and
$\coef{\comppoly{(\x_j\y_{i'})}}$ for $\x_i\y_{j'}(f^*+\sum_{M^*} c_{M^*}M^*)$
is $c_{\comppoly{(\x_i\x_j\y_{i'}\y_{j'})}}$, which is the desired result
since we assumed $c_{\comppoly{(\x_i\x_j\y_{i'}\y_{j'})}}=1$.

Requirement~\ref{item:fcs}: Suppose we have assigned a value from $\GFtwo$
to each variable $c_{M^*}$ where $M^*$ is a coquartic monomial
in $\x_1,\dots,\x_{\n},\y_1,\dots,\y_{\n'}$; we must show that this assignment
satisfies the precondition of at least one of the triples in the
new certificate.  First, assign to each coquartic monomial~$M$ in
$\x_1,\dots,\x_{\n}$ the value $c_M = c_{M\y_1\cdots\y_{\n'}}$, and
find a subproof $\langle C,q,r \rangle$ from the certificate for~$f$
such that this assignment satisfies~$C$.
Then assign to each coquartic monomial~$M'$ in
$\y_1,\dots,\y_{\n'}$ the value $c_{M'} = c_{M'\x_1\cdots\x_{\n}}$, and
find a subproof $\langle C',q',r' \rangle$ from the certificate for~$f'$
such that this assignment satisfies~$C'$.
Let $\langle C^*,q^*,r^* \rangle$ be created from $\langle C,q,r \rangle$
and $\langle C',q',r' \rangle$ as above.  If all of the coefficients
mentioned in the two new equations in~$C^*$ are assigned the value~0,
then the assignment satisfies~$C^*$; if one of these coefficients
$c_{\comppoly{(\x_i\x_j\y_{i'}\y_{j'})}}$ is assigned the value~1,
then the assignment satisfies the precondition of the
corresponding triple~\eqref{eq:newsubproof}.

So the new certificate is valid, and we are done.
\end{proof}

\end{document}